\documentclass[letterpaper, 10pt, conference]{ieeeconf}  

\usepackage{graphicx}
\usepackage[utf8]{inputenc}
\usepackage[T1]{fontenc}
\usepackage{mathtools}
\usepackage[english]{babel}
\usepackage{amsmath}
\usepackage{amssymb}
\usepackage{cleveref}
\usepackage{bm}
\usepackage{psfrag}
\usepackage{esdiff}
\usepackage{float}
\usepackage{caption}
\usepackage{subfigure}
\usepackage{bm}
\usepackage{algorithmicx}
\usepackage{algpseudocode}
\usepackage{algorithm}
\usepackage{lipsum}
\usepackage{array}
\usepackage{blindtext}
\usepackage{epstopdf}
\usepackage{ragged2e}
\usepackage{mathrsfs}
\usepackage{tabu}
\graphicspath{ {figures/} }

\usepackage{url}
\usepackage{nomencl}
\usepackage{cite}
\usepackage{lmodern}
\usepackage{enumerate}
\usepackage{tikz}

\IEEEoverridecommandlockouts                              

\newcommand{\be}{\begin{equation}}
\newcommand{\ee}{\end{equation}}
\newcommand{\bewn}{\begin{equation*}}
\newcommand{\eewn}{\end{equation*}}
\newcommand{\bbmat}{\begin{bmatrix}} 
	\newcommand{\ebmat}{\end{bmatrix}}
\newcommand{\bd}{\begin{displaymath}}
\newcommand{\ed}{\end{displaymath}}
\newcommand{\bea}{\begin{eqnarray}}
\newcommand{\eea}{\end{eqnarray}}
\newcommand{\ba}{\begin{array}}
	\newcommand{\ea}{\end{array}}
\newcommand{\baa}{\begin{array}{ll}}
	\newcommand{\eaa}{\end{array}}

\newcommand{\bc}{\begin{center}}
	\newcommand{\ec}{\end{center}}
\newcommand{\ben}{\begin{enumerate}}
	\newcommand{\een}{\end{enumerate}}
\newcommand{\bi}{\begin{itemize}}
	\newcommand{\ei}{\end{itemize}}
\newcommand{\bt}{\begin{tabular}}
	\newcommand{\et}{\end{tabular}}
\newcommand{\bte}{\begin{table}}
	\newcommand{\ete}{\end{table}}

\newcommand{\norm}[1]{\left\lVert#1\right\rVert}   

\makeatletter
\renewcommand\paragraph{\@startsection{paragraph}{4}{\z@}%
	{-2.5ex\@plus -1ex \@minus -.25ex}%
	{1.25ex \@plus .25ex}%
	{\normalfont\normalsize\bfseries}}
\makeatother



\newtheorem{theorem}{Theorem}

\newtheorem{lemma}[theorem]{\textbf{Lemma}}
\newtheorem{assumption}{\textbf{Assumption}}
\newtheorem{problem}{\textbf{Problem}}

\newcommand{\bR}{\mathbb{R}}
\newcommand{\calA}{\mathcal{A}}

\newcommand{\calC}{\mathcal{C}}
\newcommand{\calD}{\mathcal{D}}
\newcommand{\calF}{\mathcal{F}}
\newcommand{\calN}{\mathcal{N}}
\newcommand{\calO}{\mathcal{O}}
\newcommand{\calP}{\mathcal{P}}

\newcommand{\calS}{\mathcal{S}}

\newcommand{\calW}{\mathcal{W}}

\newcommand{\abs}[1]{\left |#1\right |}

\title{\LARGE \bf
Herding an Adversarial Attacker to a Safe Area for Defending Safety-Critical Infrastructure  
}

\author{Vishnu S. Chipade and Dimitra Panagou
\thanks{The authors are with the Department of Aerospace Engineering,
	University of Michigan, Ann Arbor, MI, USA;
	{\tt\small (vishnuc,
			dpanagou)@umich.edu}}
			\thanks{This work has been funded by the Center for Unmanned Aircraft Systems (C-UAS), a National Science Foundation Industry/University Cooperative Research Center (I/UCRC) under NSF Award No. 1738714 along with significant contributions from C-UAS industry members.}
}

\begin{document}

\maketitle
\thispagestyle{empty}
\pagestyle{empty}

\begin{abstract}

This paper investigates a problem of defending safety-critical infrastructure from an adversarial aerial attacker in an urban environment. A circular arc formation of defenders is formed around the attacker, and vector-field based guidance laws herd the attacker to a predefined safe area in the presence of rectangular obstacles. The defenders' formation is defined based on a novel vector field that imposes super-elliptic contours around the obstacles, to closely resemble their rectangular shape. A novel finite-time stabilizing controller is proposed to guide the defenders to their desired formation, while avoiding obstacles and inter-agent collisions. The efficacy of the approach is demonstrated via simulation results. 

\end{abstract}

\section{Introduction}

Advancements in and popularity of swarm technology pose increased risk to safety-critical infrastructure such as government facilities, airports, military bases etc. 
Even unintentional intruders such as bird flocks entering airports cause significant monetary loss~\cite{allan2000costs}. Therefore, defending safety-critical infrastructure from attacks is critically important, particularly in crowded urban areas.

Various approaches have been proposed for the herding problem, including biologically-inspired robotic herding \cite{haque2011biologically, evered2014investigation, pierson2015bio}, and game-theoretic intruder herding \cite{nardi2018game}. In \cite{licitra2017single,licitra2018single} the authors discuss herding using a switched systems approach; the herder chases targets sequentially by switching among them so that certain dwell-time conditions are satisfied to guarantee stability of the resulting trajectories. However, no input constraints were considered, while the number of chased targets is limited.  The authors in \cite{deptula2018single} use approximate dynamic programming to obtain approximately optimal control policies for the herder to chase a target agent to a goal location. However, none of these herding problems considered any obstacles in the environment. 

Similarly to our line of work, \cite{varava2017herding} uses a cage of robots to herd sheep to a goal location. An RRT approach is used to find a motion plan for the robots while maintaining the cage. The herding problem discussed in \cite{pierson2018controlling,pierson2015bio} uses a circular arc formation to influence the nonlinear dynamics of the herd based on a potential-field approach. The authors design a point-offset control to guide the herd close to a specified location. However, both papers treat robots as point masses and no inter-agent collision is addressed. In our work, we assume that the agents are of known circular footprints, and we consider the problem of inter-agent collision avoidance. The robotic herding problem discussed in \cite{gade2015herding} uses an $n$-wavefront algorithm to herd a flock of birds away from an airport, where the intruders on the boundary of the flock are influenced based on the locations of the airport and a safe area. The authors also provide stability and performance guarantees in \cite{gade2016robotic}, as well as experimental results in \cite{paranjape2018robotic}.

In this paper, we address the problem of protecting safety-critical infrastructure (called the protected area) from an adversarial agent (called attacker hereafter), in the presence of rectangular obstacles. We utilize the idea of forming a circular arc of a team of defending agents (called defenders hereafter) around the attacker, in order to influence the motion trajectories of the attacker \cite{pierson2018controlling,pierson2015bio}. The number of defenders depends on the obstacle geometry and the parameters of the attacker's repulsion law.
In addition, we consider that the defenders should converge to their desired positions in finite time to ensure the defenders' timely convergence before the attacker reaches the protected area. Finite-time specifications have been considered recently in multi-agent control problems, e.g., finite-time tracking under leader-follower setting~\cite{wen2016distributed}, finite-time containment~\cite{wang2014distributed}. For an overview on finite-time stability, the interested reader is referred to \cite{bhat2000finite}. Here we employ the idea of finite-time stabilizing controller inspired from \cite{garg2018new} to guide the defenders to their desired positions. 
We also propose a novel vector field around the rectangular obstacles to be used as reference for the formation of the defenders. The vector field has no other singular points except at the center of the safe area. As thus, it provides a safe and globally attractive motion plan for the formation of the defenders to herd the attacker to the safe area. 

Compared to the authors' earlier work \cite{panagou2014motion,panagou2017distributed} the main contributions of this paper are: (1) We consider rectangular static obstacles, which is a more realistic model for urban environments compared to circular obstacles, and design repulsive vector fields around them using super-elliptic contours. Under proper blending with attractive vector fields, we show that the resulting vector field is a globally safe motion plan for the formation of the defenders around the attacker. (2) We design a finite-time stabilizing, state-feedback controller to force the defenders to form and maintain their desired formation while avoiding inter-agent collisions.

The rest of the paper is structured as follows: Section \ref{sec:math_model} describes the mathematical modeling and problem statement. The assumptions on the attacker's control strategy are discussed in Section \ref{sec:attack_control}. The technical developments on the herding strategy and the vector-field based motion planning for the defenders are provided in Section \ref{sec:herding_strategy} and \ref{sec:vector_field}, respectively. Section \ref{sec:tracking_cotntrol} presents the finite-time tracking controller that guides the defenders to their desired formation. Convergence and safety are formally proved in Section \ref{sec:convergence}, while simulations are provided in Section \ref{sec:simulations}. The conclusions and our thoughts on future work are discussed in Section \ref{sec:conclusions}.

\section{Modeling and Problem Statement}\label{sec:math_model}
\textit{Notation}: Vectors are denoted by bold letters ($\textbf{r}$). Script letters denote sets ($\calP$). $R_{b_2}^{b_1}(t)$ and $E_{ok}^{b_1}(t)$ are the Euclidean distance between object $b_2$ and $b_1$, and the Super-elliptic distance between $b_1$ and $\calO_k$, respectively, at time $t$. The argument $t$ would be omitted whenever clear from the context.
$\sigma_{b_2}^{b_1}(\delta)$ is a blending function~\cite{panagou2014motion}, characterized by a triplet $\delta_{\sigma}=(\delta^m,\bar\delta, \delta^u)$ and defined in Eq.~\ref{eq:blendFun}, corresponding to the field around the object $b_2$ for the object $b_1$ which is at a distance $\delta$ from $b_2$. 
\be \label{eq:blendFun}
\sigma_{b_2}^{b_1}(\delta) =  
\begin{cases}
	1, & \delta^m \le \delta \le \bar \delta;	\\[3pt]
	A_{b_2}^{b_1}\delta^3+B_{b_2}^{b_1} \delta^2+ C_{b_2}^{b_1} \delta+D_{b_2}^{b_1}, &   \bar \delta \le \delta \le \delta^u;	\\
	0, &	\delta^u \le \delta;
\end{cases}
\ee
The coefficients $A_{b_2}^{b_1},B_{b_2}^{b_1},C_{b_2}^{b_1},D_{b_2}^{b_1}$ are chosen as: $A_{b_2}^{b_1}=\frac{2}{(\delta^u-\bar \delta)^3}$, $B_{b_2}^{b_1}=\frac{-3(\delta^u+\bar \delta)}{(\delta^u-\bar \delta)^3}$,  $C_{b_2}^{b_1}=\frac{6\delta^u\bar \delta}{(\delta^u-\bar \delta)^3}$, $D_{b_2}^{b_1}=\frac{(\delta^u)^2(\delta^u-3\bar \delta)}{(\delta^u-\bar \delta)^3}$, so that \eqref{eq:blendFun} is a $\calC^1$ function. The argument $\delta$ is either the Euclidean distance or the Super-elliptic distance, depending on the objects under consideration, and would be omitted whenever clear from the context. 

 We consider an environment $\calW \subseteq \mathbb{R}^2$ with $N_o$ rectangular obstacles, a protected area $\calP \subseteq \calW$ defined as $\calP=\{\textbf{r} \in \bR^2 \;| \; \norm{\textbf{r}-\textbf{r}_p}\le \rho_p\}$, and a safe area $\calS \subseteq \calW$ defined as $\calS=\{\textbf{r} \in \bR^2 \; | \; \norm{\textbf{r}-\textbf{r}_s}\le \rho_s\}$, where $(\textbf r_p, \rho_p)$ and $(\textbf r_s, \rho_s)$ are the centers and radii of the corresponding areas, respectively. An attacker $\calA$ and $N_d$ defenders $\calD_j$, $j \in I_d= \{1,2,...,N_d\}$, are operating in $\calW$. $\calA$ and $\calD_j$ are modeled as discs of radii $\rho_a$ and $\rho_d\le \rho_a$, respectively, under single integrator dynamics: 
\be \label{eq:attackDyn1}
\dot{\textbf{r}}_{a}
=\textbf{v}_{a},
\ee
\vspace{-7mm}
\be\label{eq:defendDyn1}
\dot{\textbf{r}}_{dj}
=\textbf{v}_{dj},
\ee
\vspace{-7mm}
\be\label{eq:constraints}
\baa
\norm{\textbf{v}_{a}} \le v_{max_{a}}, \;
\norm{\textbf{v}_{dj}} \le v_{max_{dj}}, 
\eaa
\ee
where $v_{max_{a}} < v_{max_{dj}}$, $\textbf{r}_{a}=[x_{a}\; y_{a}]^T$, $\textbf{r}_{dj}=[x_{dj}\; y_{dj}]^T$ are the position vectors of $\calA$ and $\calD_j$, respectively, with respect to (w.r.t.) a global inertial frame $\calF_g (\hat {\textbf{i}}, \hat {\textbf{j}})$; $\textbf{v}_{a}=[v_{x_{a}}\; v_{y_{a}}]^T$, $\textbf{v}_{dj}=[v_{x_{dj}}\; v_{y_{dj}}]^T$ are their control velocity vectors, respectively, whose norms are bounded by $v_{max_{a}}$ and $v_{max_{dj}}$. We also assume the following:
\begin{assumption} Every defender $\calD_j$ can sense the position $\textbf r_a$ of the attacker $\cal A$ once $\calA$ lies inside a circular sensing-zone $SZ=\{\mathbf r \in \mathbb{R}^2 |\; \norm{\textbf{r}-\textbf{r}_p} \le \rho_d^s\}$ around $\calP$. 
\end{assumption}

We consider static obstacles $\calO_k$ of rectangular shape, with their edges aligned with the axes of $\calF_g$,   
defined as:
\be
\calO_k= \{\mathbf r \in \mathbb{R}^2 | \abs{x-x_{ok}} \le \frac{w_{ok}}{2}, \abs{y-y_{ok}} \le \frac{h_{ok}}{2}\} ,
\ee
where $\mathbf r_{ok}=[x_{ok} \; y_{ok}]^T$ is the center, $w_{ok}$ and $h_{ok}$ are the lengths of $\calO_k$ along $\hat{\textbf{i}}$ and $\hat{\textbf{j}}$,  $ \forall k \in I_o = \{1,2,...,N_o\}$.

In order to defend $\calP$, $\calD_j$ should safely herd $\calA$ to $\calS$. To this end, we consider the following problem: 
\begin{problem}
	Find desired positions $\textbf{r}_{gj}$ and velocities $\dot{{\textbf{r}}}_{gj}$ for the formation of $\calD_j, \forall j \in I_d$, to guide $\calA$ into $\calS$ and keep it there ever after, and control actions $\mathbf v_{dj}$, $\forall j \in I_d$, such that $\calD_j$ safely converge to $\textbf{r}_{gj}$ in finite time. 
\end{problem}
In other words, $\forall j \in I_d$ find $\textbf{v}_{dj}$ such that 
	1) $\exists T_d^c, T_s> 0: \; R_{gj}^{dj}(t)=0, \; \forall t \ge T_d^c $ and $\norm{\textbf{r}_a(t)-\textbf{r}_s}\le \rho_s$, $\forall t \ge T_s$; 2) $\forall t\ge 0$ and $\forall k \in I_o$: $E_{ok}^{dj}(t)> \xi_{ok}^{dj,m}$, and $E_{ok}^{a}(t)> \xi_{ok}^{m}$, where $\xi_{ok}^{dj,m}$, $\xi_{ok}^{m}$ are the minimum allowed super-elliptic distances, respectively; 3) $\forall t\ge 0$:  $R_{dl}^{dj}(t)>R_{d}^{d,m}$, $\forall l \neq j$ and $R_{dj}^a\ge R_{d}^{a,m}$, where $R_{d}^{d,m}$, $R_{d}^{a,m}$ are the corresponding minimum allowed distances.

\section{Attacker's Control Strategy}\label{sec:attack_control}
The attacker $\calA$ aims to reach $\calP$ while avoiding $\calD_j$. 
The control action of $\calA$ is modeled using a vector-field approach~\cite{panagou2017distributed}, as follows: $\calA$ assumes that $\calD_j$ and any static obstacles that it senses within its sensing radius $\rho_a^s$ are of circular shape. Therefore, it moves along the direction prescribed by the vector field defined as:
\be \label{eq:attackVectField1}
\baa
\textbf{F}^a=&\displaystyle \prod_{k\in\calN_{ao}} ( 1- \sigma_{ok}^{a}(R_{ok}^{a})) \prod_{j\in\calN_{ad}} ( 1- \sigma_{dj}^{a}(R_{dj}^{a})) \textbf{F}_{p}^{a} \\
& +\displaystyle \sum_{k\in\calN_{ao}} \sigma_{ok}^{a}(R_{ok}^{a}) \textbf{F}_{ok}^a + \sum_{j\in\calN_{ad}} \sigma_{dj}^{a}(R_{dj}^{a}) \textbf{F}_{dj}^a, 
\eaa
\ee
where $\calN_{ad} = \{j \in I_d \; | \;  \norm{\textbf{r}_a-\textbf{r}_{dj}} \le \rho_a^s\}$ and $\calN_{ao} = \{k \in I_o \; | \;  \norm{\textbf{r}_a-\textbf{r}_{ok}} \le \rho_a^s\}$.
The vector field \eqref{eq:attackVectField1} comprises attractive term $\textbf{F}_{p}^{a}=\frac{\textbf{r}_p-\textbf{r}_a}{\norm{\textbf{r}_p-\textbf{r}_a}}$ and repulsive terms $\textbf{F}_{ok}^a=-\frac{\textbf{r}_{ok}-\textbf{r}_a}{\norm{\textbf{r}_{ok}-\textbf{r}_a}},  \; \textbf{F}_{dj}^a=-\frac{\textbf{r}_{dj}-\textbf{r}_a}{\norm{\textbf{r}_{dj}-\textbf{r}_a}},$.
$\forall k \in \calN_{ao}$ and $\forall j \in \calN_{ad}$, respectively. The blending functions $\sigma_{ok}^a(R_{ok}^{a})$, $\sigma_{dj}^a(R_{dj}^{a})$ are defined as in Eq.~\ref{eq:blendFun} with ($R_{ok}^{a,m}$, $\bar R_{ok}^{a}$, $R_{ok}^{a,u}$ ) and ($R_{d}^{a,m}$, $\bar R_{d}^{a}$, $R_{d}^{a,u}$ ) as the corresponding $\delta_{\sigma}$ triplets.
\begin{assumption}\label{assum:attack_nzero_speed}
The attacker $\calA$ has a non-zero speed when it is close to an obstacle $\calO_k$ or a defender $\calD_j$ i.e., $\norm{\textbf{v}_{a}} > 0$ when \{$\exists k \in I_o: \sigma_{ok}^a \neq 0$\}$ \vee $\{$\exists j \in I_d: \sigma_{dj}^a \neq 0$\}. 
\end{assumption}
\begin{lemma}
	In the absence of defenders $\calD_j$ in $\calW$ and if the obstacles $\calO_k$ satisfy: $\norm{\textbf{r}_{ok}-\textbf{r}_{ol}}\ge R_{ok}^{a,u}+R_{ol}^{a,u}$ for all $k, l \in I_o$ and $k\neq l$, then the vector field $\textbf{F}^a:\mathbb{R}^2\rightarrow \mathbb{R}^2$ as defined in Eq.~\eqref{eq:attackVectField1} is a safe and almost globally attractive field in $\calW_s^a = \calW 	\setminus \displaystyle \bigcup_{k \in I_o}\{\textbf{r}\in\bR^2|\norm{\textbf r-\textbf r_{ok}}<R_{ok}^{a,m}\}$ to the center $\textbf{r}_p$ of $\calP$, except for a set initial conditions of measure zero. 
\end{lemma}
\begin{proof}Refer to Theorem 4 in \cite{panagou2017distributed}.
\end{proof}
The control action of the attacker $\calA$ is given as:
\be\label{eq:attack_control}
\dot {\textbf{r}}_a=
\begin{cases}
	 \frac{\textbf{F}^a}{\norm{\textbf{F}^a}}\norm {\textbf v_a}, &  \textbf{F}^a\neq \textbf 0;\\
	  \bbmat \cos(\theta^-+\epsilon_a )\\ \sin(\theta^-+\epsilon_a ) \ebmat\norm {\textbf v_a}, & \textbf{F}^a = \textbf 0.
\end{cases}
\ee
with $\theta^- = \theta(t^-)$, where $t$ is the time when $ \textbf{F}^a = \textbf 0$ and $\epsilon_a <<1$. The attacker $\calA$ moves slightly away from the direction imposed by $\mathbf F^a$ when in the neighborhood of $\textbf{F}^a = \textbf 0$, to get out of the deadlock situation.


\section{Herding Strategy}\label{sec:herding_strategy}

We now develop the strategy to herd $\calA$ to the safe area $\calS$. We build on the work by Pierson et al.~\cite{pierson2018controlling,pierson2015bio}, in which $\calD_j$ are positioned along a circular arc formation around $\calA$ to direct it towards $\calS$ as shown in Fig.~\ref{fig:circle_form}. By doing so the defenders generate a sufficient repulsive force on $\calA$ such that $\calA$ does not aim to escape even through the gaps between the defenders.

\begin{figure}[htbp]
	\centering
	\includegraphics[width=.85\linewidth]{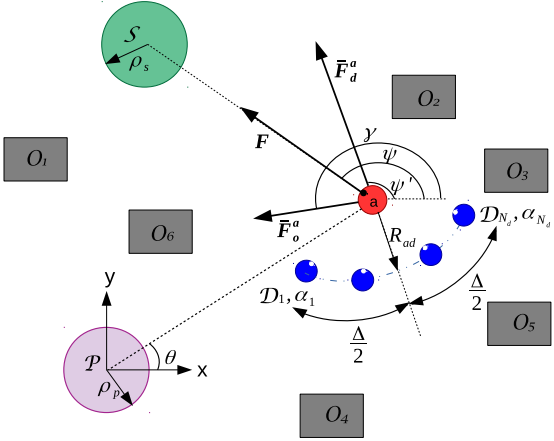}
	\caption{Circular formation around the attacker}
	\label{fig:circle_form}
\end{figure}
\vspace{-1mm}
More specifically, we place the $N_d$ defenders $\calD_j$ along the circular arc of radius $R_{ad}$ in symmetric formation, so that the angle $\alpha_j$ is chosen as: 
\be
\alpha_j=\psi'+\pi+\Delta_j, \; \text{where} \; \Delta_j=\Delta \frac{(2j-N_d-1)}{(2N_d-2)},
\ee
for some angle $\psi'$ along which the resultant repulsive vector field from the defenders is pointed, where $\Delta$ is the total spread angle of $\calD_j$'s around $\calA$. For a given number $N_d$, a minimum separation $R_{d}^{a,m}$ allowed between $\calA$ and $\calD_j$, and a minimum separation $R_{d}^{d,m} >\rho_d$ allowed among $\calD_j$, the minimum spread required is:
 \bewn \Delta_{min} = (N_d-1) \cos^{-1}\left (1-\frac{ (R_{d}^{d,m})^2}{2(R_{d}^{a,m})^2} \right). \eewn  We choose $R_{ad}=\frac{R_{d}^{a,m}+\bar R_{d}^{a}}{2} \le \bar R_{d}^{a}$.
With this choice of $R_{ad}$, we have $\sigma_{dj}^a = 1$ for all $j \in I_d$ and hence the vector field $\textbf{F}^a$ becomes:
\be \label{eq:attackVectField2}
\textbf{F}^a=\displaystyle \sum_{k\in\calN_{ao}} \sigma_{ok}^{a} \textbf{F}_{ok}^a + \sum_{j=1}^{N_d} \bbmat C(\alpha_j) \\ S(\alpha_j)\ebmat,
\ee
where $C(.)=\cos(.)$ and $S(.)=\sin(.)$. 

 Let the repulsive fields $\textbf{F}_{ok}^a$ due to all $\calO_k$ in $\calA$'s sensing radius add up to give: $\bar{\textbf{F}}_o^a = \bar F_o^a \bbmat C(\gamma) \\ S(\gamma)\ebmat$. 
 $\textbf{F}^a$ then is:
 \be \label{eq:attackVectField3}
 \textbf{F}^a= \bar{F}_{o}^a \bbmat C(\gamma) \\ S(\gamma)\ebmat + \frac{S(\frac{N_d\Delta}{2N_d-2})}{S(\frac{\Delta}{2N_d-2})} \bbmat C(\psi') \\
 S(\psi') \ebmat.
 \ee
 Let $\bar \Delta_S = \frac{S(\frac{\Delta}{2N_d-2})}{S(\frac{N_d \Delta}{2N_d-2})}$, then $\textbf{F}^a$ can be rewritten as:
 \be \label{eq:attackVectField4}
 \textbf{F}^a= \bar{F}_{o}^a \bbmat C(\gamma) \\ S(\gamma)\ebmat +  \bar \Delta_S  \bbmat C(\psi') \\
 S(\psi') \ebmat.
 \ee
 To guide $\calA$ along a direction $\psi$ of our choice, we need:
 \be \label{eq:attack_vector_field5}
 \frac{\textbf{F}^a}{\norm{\textbf{F}^a}} = \bbmat C(\psi) \\ S(\psi) \ebmat.
 \ee
 This is achieved by choosing $\psi'$ such that:   
 \be \label{eq:desired_psi_dash}
 \bar \Delta_S \left ( S(\psi')-T(\psi)C(\psi') \right )=\bar F_o^a \left (T(\psi)C(\gamma) - S(\gamma) \right ),
 \ee
 \be \label{eq:desired_psi_dash_dot}
 \baa
 \dot \psi'=&\frac{\dot {\bar F}_o^a (T(\psi)C(\gamma)-S (\gamma))-\bar F_o^a \dot {\gamma}(T(\psi)S(\gamma)+C(\gamma))}{\bar \Delta_S(C(\psi')+T(\psi)S(\psi'))}\\
 &+\frac{\sec^2(\psi) \dot {\psi}\left(\bar F_o^a C(\gamma)+\bar\Delta_S(C(\psi')\right)}{\bar \Delta_S(C(\psi')+T(\psi)S(\psi'))}.
 \eaa
 \ee
We choose $\Delta$ and $N_d$ such that $\bar \Delta_S > \bar F_o^a$ and hence $\dot \psi'$ is always finite and bounded by $\dot{\psi}_{max}'$.

 The desired (goal) position and velocity of $\calD_j$ in the circular arc formation are then obtained as:
\be \label{eq:desiredDefPos}
\textbf{r}_{gj}=\textbf{r}_a+R_{ad}\bbmat C(\psi'+\pi+\Delta_j)\\ S(\psi'+\pi+\Delta_j)\ebmat,
\ee 
\be \label{eq:desiredDefVel}
\dot {\textbf{r}}_{gj}=\dot {\textbf{r}}_a +  R_{ad}  \dot \psi' \bbmat -S(\psi'+\pi+\Delta_j)\\ C(\psi'+\pi+\Delta_j)\ebmat.
\ee

\section{Guiding the formation in an environment with obstacles}\label{sec:vector_field}

In this section, we design the reference heading $\psi$ for the formation of the attacker $\calA$ and the defenders $\calD_j$ such that the obstacles $\calO_k$ in the environment are avoided. The formation is centered at the location of the attacker $\mathbf r_a$ and has radius $\rho_f=R_{ad}+\rho_d$. To avoid $\calO_k$ we define a vector field that closely follows the rectangular shape. We use the idea of superquadric isopotential contours \cite{volpe1990manipulator}, which are defined in Eq.~\eqref{eq:superEllDist1} for different values of super-elliptic distance $E_{ok}$: 
\be\label{eq:superEllDist1}
E_{ok} = \abs{\frac{x-x_{ok}}{a_{ok}}}^{2n_{ok}} + \abs{\frac{y-y_{ok}}{b_{ok}}}^{2n_{ok}} -1,
\ee
where $a_{ok}=\frac{ w_{ok}}{2}(2)^{\frac{1}{2n_{ok}}}$ and $b_{ok}=\frac{ h_{ok}}{2}(2)^{\frac{1}{2n_{ok}}}$.
The shape of the contour corresponding to a fixed value of $E_{ok}$ becomes more rectangular as $n_{ok}$ tends to $\infty$, i.e. by choosing large $n_{ok}$ we can ensure, the contours closely match the shape of the rectangular building, see Fig.\ref{fig:vectFieldDef}.

To account for the safe distance ($R_{safe}$) from the obstacle boundary and the maximum size of the formation ($\rho_f$), the rectangular obstacles are inflated as:
\be
\calO_k^i= \left \{\mathbf r \in \mathbb{R}^2 | \abs{x-x_{ok}} \le \frac{ \bar w_{ok}}{2}, \; \abs{y-y_{ok}} \le \frac{ \bar h_{ok}}{2}\right \} ,
\ee
where $\bar w_{ok}=w_{ok}+2(\rho_f+R_{safe})$ and $\bar h_{ok}=h_{ok}+2(\rho_f+R_{safe})$ for all $ k \in I_o$.
We choose constant values for $n_{ok}$: 
\be
\baa \label{eq:superQuadCont2}

n_{ok}&=\frac{1}{1-e^{-\xi_{ok}^m}} \;,  \; \xi_{ok}^m=\frac{1}{2}\left(\abs{\frac{\bar w_{ok}}{w_{ok}}}^{2n_{ok}} + \abs{\frac{\bar h_{ok}}{h_{ok}}}^{2n_{ok}} \right)-1 .
\eaa
\ee
The inflated rectangular obstacles can be approximated by a bounding super-ellipse, given as:
\be
\calO_k^m= \left \{ \textbf{r} \in \mathbb{R}^2 | E_{ok}\le \xi_{ok}^m \right \}.
\ee
 The slope of the tangent to these contours at a given point $f$ with position vector $\textbf{r}_f\in\mathbb R^2$ is obtained as:
 	\be \label{eq:contour_tangent_slope}
	\tan(\bar\beta_{ok}^{f}) = -\frac{b_{ok}^{2n_{ok}}C(\beta_{ok}^{f})(C(\beta_{ok}^{f}))^{2n_{ok}-2}}{a_{ok}^{2n_{ok}}S(\beta_{ok}^{f})(S(\beta_{ok}^{f}))^{2n_{ok}-2}}
	\ee
where $\beta_{ok}^{f}$ is the angle made by the position vector of $f$ with respect to the center of $\calO_k$ and $\hat {\textbf{i}}$, given as: $\beta_{ok}^{f}=\tan^{-1}\left(\frac{y_f-y_{ok}}{x_f-x_{ok}}\right)$.
\vspace{0mm}
\subsection{Repulsive vector field}
The repulsive vector field around $\calO_k^m$ is defined as: 
\vspace{-1mm}
\be \label{eq:repulsField1}
\textbf{F}_{ok}^{f}= \bbmat \cos(\phi_{ok}^{f,s})\\\sin(\phi_{ok}^{f,s}) \ebmat, 
\ee
\vspace{-1mm}
where
\be \label{eq:fieldAngle}
\phi_{ok}^{f,s}=
\begin{cases}
	\bar\beta_{ok}^{f}-\Delta\bar\beta_{ok}^{s}+\frac{\Delta\beta_{ok}^{f,s}}{\pi}(\Delta\bar\beta_{ok}^{s}-\pi), &\Delta\beta_{ok}^{f,s}< \pi ;\\
	\bar\beta_{ok}^{f}-\Delta\bar\beta_{ok}^{s}(\frac{\Delta\beta_{ok}^{f,s}-\pi}{\pi}), & \text{otherwise};
\end{cases}
\ee
where $\Delta\beta_{ok}^{f,s}\in [0,2\pi]$ is the angle between the vectors $\textbf{r}-\textbf{r}_{ok}$ and $\textbf{r}_{s} -\textbf{r}_{ok}$, defined as: $\Delta\beta_{ok}^{f,s}=\beta_{ok}^{f}-\beta_{ok}^{s}$,
and, similarly, $\Delta\bar\beta_{ok}^{s} \in [0,2\pi]$ is the angle between the tangent to the contour corresponding to the obstacle $\calO_k$ at the center of $\calS$ (i.e., $\bar\beta_{ok}^{s}$) and the vector $\textbf{r}_s-\textbf{r}_{ok}$, shown with pink text in Fig.~\ref{fig:vectFieldDef}, defined as: $\Delta\bar\beta_{ok}^{s}=\bar\beta_{ok}^{s}-\beta_{ok}^{s}$.
\vspace{-2mm}
\begin{figure}[htbp]
	\centering
	\includegraphics[width=.95\linewidth]{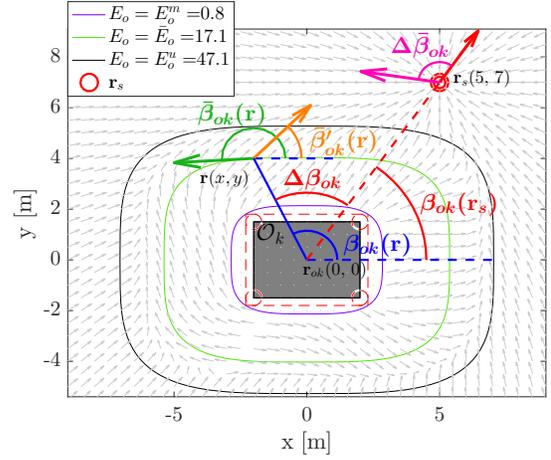}
	\caption{Vector field definition around a rectangle }
	\label{fig:vectFieldDef}
\end{figure}

\vspace{0mm}
\subsection{Attractive vector field}
A radially converging vector field is considered around the centre of $\calS$, defined as:
\be \label{eq:attractField1}
\baa
\textbf{F}_{s}^{f}=\frac{\textbf{r}_s-\textbf{r}_f}{\norm{\textbf{r}_s-\textbf{r}_f}}.
\eaa
\ee
We set $\textbf{F}_{s}^{f}=\mathbf 0$ at $\textbf{r}_{f}=\bm{0}$, so $\textbf{F}_{s}^{f}$ is defined everywhere. 
\subsection{Combining attractive and repulsive vector fields}
In order to combine the attractive and repulsive vector fields smoothly near $\calO_k$, a blending function $\sigma_{ok}^f (E_{ok}^{f})$ is defined as in Eq.~\ref{eq:blendFun} with ($\xi_{ok}^{m}$, $\bar \xi_{ok}$, $\xi_{ok}^{u}$) as the corresponding $\delta_{\sigma}$ triplet.
This essentially means that the repulsive effect of $\calO_k$ is confined within $(\calO_k^m)^c\cap(\calO_k^u)$, where $(\calO_k^m)^c$ is the compliment of $\calO_k^m$ and $\calO_k^u$ is given as:
\be
\calO_k^u= \left \{ \textbf{r} \in \mathbb{R}^2 \; | \; E_{ok} \le \xi_{ok}^u \right \}.
\ee

The attractive and repulsive terms are combined as:
\be \label{eq:combVectField1}
\textbf{F}^{f}=\displaystyle \prod_{k=1}^{N_o} ( 1- \sigma_{ok}^f (E_{ok}^f)) \textbf{F}_s^{f} + \displaystyle \sum_{k=1}^{N_o} \sigma_{ok}^f  (E_{ok}^f) \textbf{F}_{ok}^{f}.
\ee
 The vector field $\textbf{F}^f$ around $\calO_k$ with $(w_{ok},h_{ok})=(4,3)$, $\textbf{r}_{ok}=[0, 0]^T$; and $\calS$ with $\textbf{r}_s=[3,7]^T$, is shown in Fig.~\ref{fig:vectFieldDef}.
 \vspace{-2mm}
\begin{assumption}\label{assum:obstacle_field}
	The $N_o$ rectangular obstacles $\calO_k$, $k \in I_o$, are spaced such that $\displaystyle \bigcap_{k \in I_o}\calO_k^u = \emptyset$ and $\calS \notin \displaystyle \bigcup_{k \in I_o}\calO_k^u$, i.e., at any location, the effect of at most one obstacle is active, and there is no effect of any obstacle in the safe area $\calS$. 
\end{assumption}
\begin{theorem}\label{th:safe-convergent_motion_plan}
If Assumption \ref{assum:obstacle_field} holds, then the vector field $\textbf{F}^{f}:\mathbb{R}^2 \rightarrow \mathbb{R}^2$ as defined in Eq.~\eqref{eq:combVectField1} is a safe, and globally attractive motion plan to the safe area $\calS$ centered at $\textbf{r}_s$ in $\calW_s = \calW 	\setminus \displaystyle \bigcup_{k \in I_o}\calO_k^m$.
\end{theorem}
\begin{proof}
First, we show that $\textbf{F}^{f}$ in $\calW_s$ has no singular point other than $\textbf{r}_f=\textbf{r}_s$, to ensure that no initial conditions in $\calW_s$ yields trajectories that end at the singular point. To verify this, we check $\norm{\textbf{F}^{f}}$ in 3 disjoint subsets $\calW_{s1}$, $\calW_{s2}$ and $\calW_{s3}$ of $\calW_{s}$ under assumption \ref{assum:obstacle_field}. Let $\sigma_{ok}^f=\sigma_{ok}^f(E_{ok}(\textbf{r}_f))$.\\
$\bullet$
 Set $\calW_{s1}=\{\textbf{r}_f\in \calW_s\;|\;\sigma_{ok}^f=0, \; \forall k \in I_o\}$. So, we have, $\textbf{F}^{f}=\textbf{F}_{s}^{f}$ and $\norm{\textbf{F}^{f}}=1$ for all $\textbf{r}_f\neq \textbf{r}_s \in \calW_{s1}$. 
 \\$\bullet$
Set $\calW_{s2}=\displaystyle \bigcup_{k \in I_o} \{\textbf{r}_f\in \calW_s\;|\;\sigma_{ok}^f=1\}$. So, $\exists !$ $k \in I_o$ such that $\textbf{F}^{f}=\textbf{F}_{ok}^{f}$ and $\norm{\textbf{F}^{f}}=1$ for all $\textbf{r}_f \in \calW_{s2}$.
 \\$\bullet$
 Set $\calW_{s3}=\displaystyle \bigcup_{k \in I_o} \{\textbf{r}_f\in \calW_s \;|\; 0<\sigma_{ok}^f<1\}$. For any $\textbf{r}_f\in \calW_{s3}$, $\exists!$ $k \in I_o$ such that $\textbf{F}^{f}$ is given as: $
 \textbf{F}^{f}=(1-\sigma_{ok}^f)\textbf{F}_{s}^{f}+\sigma_{ok}^f\textbf{F}_{ok}^{f},
 $
 with the norm defined as:
 \be
 \norm{\textbf{F}^{f}}=\sqrt{1+2\sigma_{ok}^f\left(1-\sigma_{ok}^f\right)\left (\cos(\partial{\bar\beta}_{ok})-1 \right )},
 \ee
 where $\partial{\bar\beta}_{ok}$ is the angle between the two vectors $\textbf{F}_{s}^{f}$ and $\textbf{F}_{ok}^{f}$.
For $\sigma_{ok}^f \in (0,1)$, we have $0<\sigma_{ok}^f(1-\sigma_{ok}^f)<\frac{1}{4}$. 
Now, if $\cos(\partial{\bar\beta}_{ok})\in (-1,1]$, then $-1<2\sigma_{ok}^f\left(1-\sigma_{ok}^f\right)\left (\cos(\partial{\bar\beta}_{ok})-1 \right )\le0$. This implies $\norm{\textbf{F}^{f}}>0$.
Then the only condition for $\norm{\textbf{F}^{f}}$ to be 0 is to have $\cos(\partial{\bar\beta}_{ok})=-1$, i.e. $\partial{\bar\beta}_{ok}=\pm\; \pi$ and $\sigma_{ok}^f=0.5$.
Now, we will show that $\partial{\bar\beta}_{ok}$ can never become $\pm \;\pi$. To see that, let us consider $\Delta\beta_{ok}^{f,s}\ge \pi$, then $\partial{\bar\beta}_{ok}$ becomes: 
\vspace{0mm}
\be  \label{eq:delBeta1}
\partial{\bar\beta}_{ok} = \tan^{-1}\left(
\frac{y_s-y_f}{x_s-x_f}\right)-\bar\beta_{ok}^{f}+\Delta\bar\beta_{ok}^{s}\left(\frac{\Delta\beta_{ok}^{f,s}-\pi}{\pi}\right).
\ee
For a given $\textbf{r}_s$, the angle $\phi_{ok}^{f,s}$ of the vector field $\textbf{F}_{ok}^{f}$ is only dependent on the angle $\beta_{ok}^{f}$. Also for a given $\beta_{ok}^{f}$, the angle $\tan^{-1}\left(
\frac{y_s-y_f}{x_s-x_f}\right)$ becomes larger with larger $\norm{\textbf{r}_f-\textbf{r}_{ok}}$. So, let us consider a worst case scenario where the contour corresponding to $E_{ok}=\xi_{ok}^u$ just touches the safe location $\textbf{r}_s$. 
 After parameterizing this contour in terms of $\beta_{ok}^{f}$, we get $x_f=x_{ok}+\bar a_{ok} (\cos(p_{ok}))^{\frac{1}{n_{ok}}}$ and $y_f=y_{ok}+\bar b_{ok}(\sin(p_{ok}))^{\frac{1}{n_{ok}}}$
where $p_{ok}=\tan^{-1} \left(\frac{\left(\bar a_{ok}\sin(\beta_{ok}^{f})\right)^{n_{ok}}}{\left(\bar b_{ok}\cos(\beta_{ok}^{f})\right)^{n_{ok}}}\right)$, $\bar a_{ok}=a_{ok}(1+\xi_{ok}^u)^\frac{1}{2n_{ok}}$, and $\bar b_{ok}=b_{ok}(1+\xi_{ok}^u)^\frac{1}{2n_{ok}}$.
Because of the symmetry, we only consider $\calS$ with $\textbf{r}_s$ in the first quadrant w.r.t $\textbf{r}_{ok}$. The function $\partial{\bar\beta}_{ok}$ in Eq.~\eqref{eq:delBeta1} obtained in terms of $\beta_{ok}^{f}$, after replacing $x_f$, $y_f$ and other variables, is continuous in its arguments, but it is not trivial to find the maximum value of $\partial{\bar\beta}_{ok}$. So we find the maximum and minimum $\partial{\bar\beta}_{ok}$ for $\Delta\beta_{ok}^{f,s} \in[ \pi ,2\pi]$ for a given choice of $\beta_{ok}^{s}$. We plot the maxi and min values of $\partial{\bar\beta}_{ok}$ for different choices of  $\beta_{ok}^{s} \in [0,\frac{\pi}{2}]$ in Fig.~\ref{fig:minMaxDelBeta} and it can be observed that the value never becomes $\pi$ or $-\pi$, in fact it is far from being $\pi$ or $-\pi$.
\vspace{0mm}
\begin{figure}[htbp]
	\centering
	\includegraphics[width=.85\linewidth]{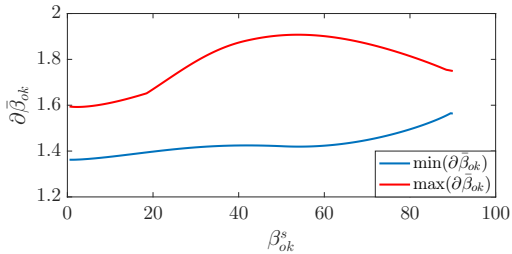}
	\caption{Minimum and Maximum values of $\partial{\bar\beta}_{ok}$}
	\label{fig:minMaxDelBeta}
 \end{figure}
Similar analysis can be performed when $\Delta\beta_{ok}^{f,s} < \pi$ to show that $\partial{\bar\beta}_{ok}$ never becomes $\pi$ or $-\pi$, which implies that $\norm{\textbf{F}^{f}} \neq 0$ in $\calW_{s3}$ and hence there is no singular point in $\calW_s$ other than $\textbf{r}_f=\textbf{r}_s$. Now, similar to the analysis of Lemma 1 in \cite{panagou2017distributed}, it can be shown that the integral curves which enter the set $\calO_k^u$ escape the boundary of $\calO_k^u$. Also, outside $\left\{\displaystyle \bigcup_{k \in I_o}\calO_k^u\right\}$, by definition the vector field is convergent to $\textbf{r}_s$. This proves that the vector field $\textbf{F}^{f}$ is a safe and globally attractive motion plan to the safe area $\calS$.
\end{proof}

\subsection{Safe Area ($\calS$): A Forward Invariant Set}\label{subsec:safe_area_psi}
Once $\calA$ is in the safe area $\calS$, which is chosen far away from the protected area $\calP$, we gradually change the formation heading $\psi$ to a direction which is nearly perpendicular to the direction given by the vector field $\textbf{F}^{f}$, and keep applying the nearly perpendicular direction thereafter. This traps $\calA$ inside $\calS$, while yielding smooth controls for $\calD_j$. Therefore $\psi$ is designed as:
 \be \label{eq:formationHeadAngle2}
\psi =
\begin{cases}
	\tan^{-1}\left( \frac{F_y}{F_x} \right ), &  \norm{\textbf{r}_a-\textbf{r}_s} > \rho_s;\\
	\tan^{-1}\left( \frac{F_y}{F_x}\right)+\int_{T_s}^{t} \frac{(\frac{\pi}{2}-\epsilon_s)}{\Delta T_t} d\tau,& T_s \le t\le T_s+\Delta T_t;\\
		\tan^{-1}\left( \frac{F_y}{F_x}\right)+(\frac{\pi}{2}-\epsilon_s), &  T_s+T_t \le  t;
\end{cases}
\ee
where $T_s$ is the time when $\calA$ enters $\calS$, $\Delta T_t$ is a transition time, and $\epsilon_s$ is a small number. Let $v_{max_d}=\displaystyle \min_{j \in I_d} v_{max_{dj}}$.  $\Delta T_t$ and $\rho_s$ are chosen as:  $\Delta T_t \ge \frac{\pi}{2}\frac{\left(v_{max_d}-v_{max_a}\right)}{R_{ad}}$ and $\rho_s\ge v_{max_a}\Delta T_t+\frac{v_{max_a}R_{ad}}{v_{max_d}-v_{max_a}}$.  
\section{Reference Trajectory Tracking}\label{sec:tracking_cotntrol}
To ensure that the attacker $\calA$ is herded along the $\psi$ direction, the defenders $\calD_j$'s have to track $\textbf{r}_{gj}$ (Eq. \eqref{eq:desiredDefPos} and $\dot{\textbf{r}}_{gj}$  (Eq.~\eqref{eq:desiredDefVel}), while avoiding collisions with $\calO_k$'s and $\calD_l$'s.
To ensure that $\textbf{r}_{gj}$  does not intersect any $\calO_k$ for all $j \in I_d$, the following assumptions are made.
\begin{assumption} \label{assum:R_m_aok}
	 $R_{ok}^{a,m} \ge \sqrt{\bar w_{ok}^2+\bar h_{ok}^2}$, $\forall k \in I_o$.
\end{assumption}
\begin{assumption} \label{assum:R_safe}
	$R_{safe} \ge 2\rho_d$. 
\end{assumption}

\subsection{Repulsive vector field from the obstacles for $\calD_j$}
The repulsive vector field at $\textbf{r}_{dj}=[x_{dj} \; y_{dj}]^T$ around the obstacle $\calO_k$ is defined as: 
\be \label{eq:defendRepulsField1}
\baa
\textbf{F}_{ok}^{dj}= \bbmat \cos(\phi_{ok}^{dj,gj})\\\sin(\phi_{ok}^{dj,gj})\ \ebmat, \\
\eaa
\ee
where $\phi_{ok}^{dj,gj}$ is defined similar to $\phi_{ok}^{f,s}$ as in Eq.~\ref{eq:fieldAngle}.

\subsection{Repulsive vector fields from $\calD_l$ for $\calD_j$}
A radially diverging vector field is considered around each one of the defenders $\calD_l$, defined as:
\be \label{eq:defendRepulsField2}
\baa
\textbf{F}_{dl}^{dj}=-\frac{\textbf{r}_{dl}-\textbf{r}_{dj}}{\norm{\textbf{r}_{dl}-\textbf{r}_{dj}}},  \quad \forall l \neq j \in I_d.
\eaa
\ee

\subsection{Attractive vector field to the desired position of $\calD_j$}
A radially converging vector field is considered around the desired position of the defender $\textbf{r}_{gj}$, defined as:
\be \label{eq:defendAttractField1}
\baa
\textbf{F}_{gj}^{dj}=\frac{\textbf{r}_{gj}-\textbf{r}_{dj}}{\norm{\textbf{r}_{gj}-\textbf{r}_{dj}}}.
\eaa
\ee

\subsection{Resultant vector field for $\calD_j$}
In order to combine the attractive and repulsive vector fields for $\calD_j$ smoothly near the obstacles $\calO_k$ and other defenders $\calD_l$, blending functions $\sigma_{ok}^{dj} (E_{ok}^{dj})$ and $\sigma_{dl}^{dj}(R_{dl}^{dj})$ are defined respectively using Eq.~\ref{eq:blendFun} with ($\xi_{ok}^{dj,m},\bar \xi_{ok}^{dj}, \xi_{ok}^{dj,u}$) and ($R_{d}^{d,m}, \bar R_{d}^{d}, R_{d}^{d,u}$) as the corresponding $\delta_\sigma$ triplets.  

The combined vector field for $\calD_j$ is then:
\be \label{eq:defendCombVectField1}
\baa
\textbf{F}^{dj}=&\displaystyle \prod_{k=1}^{N_o} ( 1- \sigma_{ok}^{dj}(E_{ok}^{dj})) \displaystyle \prod_{l\neq j \in I_d} ( 1- \sigma_{dl}^{dj}(E_{ok}^{dj})) \textbf{F}_{gj}^{dj} \\
&+ \displaystyle \sum_{k=1}^{N_o} \sigma_{ok}^{dj} (E_{ok}^{dj}) \textbf{F}_{ok}^{dj} + \displaystyle \sum_{l\neq j \in I_d} \sigma_{dl}^{dj} (R_{dl}^{dj}) \textbf{F}_{dl}^{dj}.
\eaa
\ee
Let us define $\textbf{e}_{dj}=\textbf{r}_{dj}-\textbf{r}_{gj}$, and $\mathbf v_{dj}^0$ as:

\be \label{eq:defendCont0}
\textbf{v}_{dj}^0 =
\begin{cases}
	 k_{d0} \tanh\left(\norm{\textbf{e}_{dj}}\right)\frac{\textbf{F}^{dj}}{\norm{\textbf{F}^{dj}}}, & \norm{\textbf{e}_{dj}} > e_{dj}^t ;\\ 
	k_{d1}\norm{\textbf{e}_{dj}}^{k_{d2}}\frac{\textbf{F}^{dj}}{\norm{\textbf{F}^{dj}}}, & \norm{\textbf{e}_{dj}} \le e_{dj}^t;
\end{cases}
\ee
where $0<k_{d2}<1$ and $k_{d0}=v_{max_{dj}}-v_{max_a}-R_{ad}\dot\psi_{max}'$ to ensure bounded $\mathbf v_{dj}^0$. To ensure continuity of $\mathbf v_{dj}^0$ at $\norm{\textbf{e}_{dj}}=e_{dj}^t$, the parameters $k_{d1}$ and $e_{dj}^t$, are chosen as: $
k_{d1}=k_{d0}\frac{\tanh(e_{dj}^t)}{(e_{dj}^t)^{k_{d2}}}, \;
(1-\tanh^2(e_{dj}^t))= k_{d2}\frac{\tanh(e_{dj}^t)}{e_{dj}^t}.
$

The control action for $\calD_j$ is then designed as:
\vspace{-2mm}
\be \label{eq:defendCont1}
\textbf{v}_{dj} =
\begin{cases}
\dot{\textbf{r}}_{gj} + \mathbf v_{dj}^0, &  \;   \sigma_{ok}^{dj} = \sigma_{dl}^{dj} =0, \forall k, l\neq j ; \\ 
\mathbf v_{dj}^0, &  \text{otherwise}.
\end{cases}
\ee

Under $\textbf{v}_{dj}$, $\calD_j$ tracks the reference trajectory ${\textbf{r}}_{gj}$ whenever it is not in conflict with any $\calO_k$ or $\calD_l, l\neq j$, otherwise it only tries to get closer to $\textbf{r}_{gj}$ by moving along the vector field $\textbf{F}^{dj}$ until it resolves the conflict.   

%

\section{Convergence and Safety Analysis}\label{sec:convergence}

\begin{theorem}\label{th:defender_convergence}
	The system trajectories of Eq.~\eqref{eq:defendDyn1} converge to piecewise continuous trajectories ${\textbf{r}}_{gj}$ ($\dot{{\textbf{r}}}_{gj}$), almost globally and in finite time, with $\mathbf{v}_{dj}$ as in Eq.~\eqref{eq:defendCont1}.
\end{theorem}
\begin{proof}With Assumptions \ref{assum:R_m_aok} and \ref{assum:R_safe}, we have that $\mathbf{r}_{gj}$ never intersect the obstacles; thus one can choose some $\xi_{ok}^{dj,m}>0$ to generate $\textbf{F}_{ok}^{dj}$ for safe motion of $\calD_j$. Each defender $\calD_j$ moves under the vector field $\mathbf F^{dj}$ given by Eq.~\eqref{eq:defendCombVectField1}. By following the analysis in Theorem 5 of \cite{panagou2017distributed}, $\calD_j$'s resolve the conflicts with other $\calD_l$'s for all initial conditions except a set of initial conditions of measure zero. With finite number of $\calO$'s and $\calD$'s, $\calD_j$ resolves conflicts in finite time $T_{dj}^{0}$.
	
	Now, when $\calD_j$ is not in conflict with $\calD_l$/$\calO_k$, we have:
		\be \label{eq:def_err_dyn}
		\dot{\textbf{e}}_{dj}=
		\begin{cases}
			-k_{d0} \tanh\left(\norm{\textbf{e}_{dj}}\right)\frac{\textbf{e}^{dj}}{\norm{\textbf{e}^{dj}}}, & \quad \norm{\textbf{e}_{dj}} > e_{dj}^t; \\ 
			-k_{d1}\norm{\textbf{e}_{dj}}^{k_{d2}-1}\textbf{e}^{dj}, & \quad \norm{\textbf{e}_{dj}} \le e_{dj}^t. 
		\end{cases}
	\ee
	    Case 1: When $\norm{\textbf{e}_{dj}(T_{dj}^{0})}=e_{dj0} \le e_{dj}^t$, the dynamics read: $\dot{\textbf{e}}_{dj}=	-k_{d1}\norm{\textbf{e}_{dj}}^{k_{d2}-1}\textbf{e}^{dj}$. Then, from Lemma 1 in \cite{garg2018new}, the origin $\textbf{e}_{dj}=\textbf{0}$ is a finite-time stable equilibrium of the system \eqref{eq:def_err_dyn}. Denote the time required for convergence when $e_{dj0}=e_{dj}^t$ as $T_{dj}^f < \infty$.
	    
	    \noindent Case 2: When $e_{dj0} > e_{dj}^t$, the dynamics of the system read: $
	    \dot{\textbf{e}}_{dj}=-k_{d0} \tanh\left(\norm{\textbf{e}_{dj}}\right)\frac{\textbf{e}^{dj}}{\norm{\textbf{e}^{dj}}}
	    $.
	    From Lemma 4 in \cite{garg2019finite} there exist a finite time $T_{dj}^t=\frac{-e_{dj0}}{\tanh(e_{dj0})}\log \left(\frac{(e_{dj}^t)^2}{2 V(e_{dj0})}\right)$ such that $\norm{\textbf{e}_{dj}(t)} \le e_{dj}^t$ for all $t\ge T_{dj}^t$. This combined with case 1 implies that for $e_{dj0} > e_{dj}^t$, $\exists T_{dj}^{c}=T_{dj}^t+T_{dj}^f +T_{dj}^0< \infty$ s.t $\norm{\textbf{e}_{dj}(t)}=0$ for all $t\ge T_{dj}^c$. Hence $\textbf{r}_{dj}$ converges to $\textbf{r}_{gj}$ in finite time. 
\end{proof}
To check whether $\calD_j$ converges to $r_{gj}$ before $\calA$ reaches $\calP$, we find a conservative lower bound on the time required by $\calA$ to reach $\calP$ in the absence of obstacles and defenders as: $T_{a}^c=\frac{\norm{\textbf{r}_a(0)-\textbf{r}_p}}{v_{max_a}}$. If  $ \forall j \in I_d, \; T_{dj}^{c}<T_a^{c}$, then $\calD_j$ can herd $\calA$ to $\calS$. Finding $T_{dj}^0$ and a tight bound on $T_{a}^c$ is an involved task and will be addressed in the future work.

\begin{theorem}
If $\forall j \in I_d$ the desired positions $\textbf{r}_{gj}$ of $\calD_j$ are chosen as given in Eq.~\eqref{eq:desiredDefPos}, with $\dot{\textbf{r}}_{gj}$ as in Eq.~\eqref{eq:desiredDefVel}, where $\psi$ is given by Eq.~\eqref{eq:formationHeadAngle2}, $\calD_j$ uses the control action $\textbf{v}_{dj}$ (Eq.~\eqref{eq:defendCont1}), and $\norm{\textbf{r}_a(t)-\textbf{r}_p} >\rho_p$ at $t=T_d^c=\max_{j\in I_d} T_{dj}^c$, then $\exists T_s>0$ s.t. $\norm{\textbf{r}_a(t)-\textbf{r}_s} <\rho_s, \forall t\ge T_s$, i.e., $\calA$ is taken into $\calS$ and is trapped there forever. 
\end{theorem}
\begin{proof}
	From Theorem 3, $\calD_j$ converge to $\textbf{r}_{gj}$ in finite time $T_d^c$ for all $j \in I_d$. If $\calA$ has not reached inside $\calP$, using Eqs.~\eqref{eq:attack_control}, \eqref{eq:attack_vector_field5} and \eqref{eq:formationHeadAngle2}, the dynamics of $\calA$ reads:
	\vspace{-2mm}
	\be
	  \dot{\textbf{r}}_a= \frac{\textbf{F}^a}{\norm{\textbf{F}^a}}\norm{\textbf{v}_a}=\frac{\textbf{F}}{\norm{\textbf{F}}}\norm{\textbf{v}_a}.
	\ee
	Since from Theorem \ref{th:safe-convergent_motion_plan}, we have that the the vector field $\textbf{F}$ is non-vanishing at all points except $\textbf{r}=\textbf{r}_s$, it is safe w.r.t the obstacles and is globally convergent to $\textbf{r}_s$. In other words, there exist an integral curve from $\textbf{r}_a(T_d^c)$ which reaches $\textbf{r}=\textbf{r}_s$. Also from the Assumption \ref{assum:attack_nzero_speed}, we have that the speed $\norm{\textbf{v}_a}\neq 0 $ for all $t \ge T_d^c$. This implies that for any $\textbf{r}_a(T_d^c) \in \calW_s^a \subseteq \calW_s, \;  \exists T_s\ge T_d^c$ s.t. $\norm{r_a-r_s}=\rho_s$, i.e, $\textbf{r}_a$ enters $\calS$. Once inside $\calS$, the reference direction always points in the interior of the circle $\norm{\textbf{r}-\textbf{r}_s}=\norm{\textbf{r}_a(t)-\textbf{r}_s}$ and the choice of  $\rho_s$ for a given $v_{max_a}$, as discussed in Section \ref{subsec:safe_area_psi}, ensures that $\calA $ never escapes $\calS$. This completes the proof.
\end{proof}

\section{Simulations and Results}\label{sec:simulations}
The proposed herding strategy is evaluated via simulation results. We consider $N_o=6$ rectangular obstacles in $\calW=\bR^2$ with properties as: $(x_{ok},y_{ok},w_{ok},h_{ok})= \left\{(10,23,2,3),\;(-6,18,3,4),\;(11,5,2,2),\;(15,43,3,3),\;\right.$ $\left. (-2,45,3,4),\;(12,60,4,3)\right\}$. $\calP$ is centered at $[0,0]^T$ with $\rho_p =2 $ (pink dotted circle in Fig.~\ref{fig:herdWithObst1}) and $\calS$ is centered at $\textbf{r}_s=[-5,\;60]^T$ with $\rho_s = 5$ (green dotted circle in Fig.~\ref{fig:herdWithObst1}). $\calA, \; \calD_1,\; \calD_2$ and $\calD_3$ are initially located at $[20,48]^T$, $[-10,16]^T$, $[6,2]^T$, $[-5,-1]^T$, respectively. The other parameters are chosen as: $\rho_a =\rho_d=0.1$, $\rho_a^c=10$, $R_{safe}=2\rho_d$, $R_{ad}=0.55$, $R_{d}^{a,m}=0.3$, $R_{d}^{a,u}=0.65$.
%
%
%

Fig. \ref{fig:herdWithObst1} shows the trajectories of $\calA$ (in red) and all $\calD_j$ (in blue). It can be observed that all $\calD_j$ safely guide $\calA$ into the safe area $\calS$ and keep it there forever. $\calA$ is assumed to move at its maximum speed $v_{max_a}$ in the simulation.
\begin{figure}[!htbp]
	\centering
	\includegraphics[width=.75\linewidth]{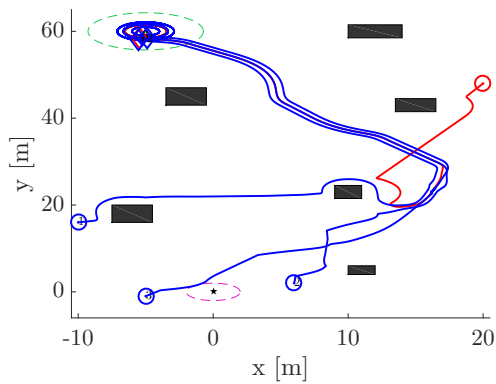}
	\caption{Trajectories of the agents during the herding}
	\label{fig:herdWithObst1}
\end{figure}
To assess the safety of the agents from $\calO_k$'s, we define critical relative distances:
\be
E_{rel}^{ao} = \displaystyle \max_{k \in I_o}\frac{\xi_{ok}^{m}}{ E_{ok}^a} , \quad 
E_{rel}^{do} = \displaystyle \max_{k \in I_o} \displaystyle \max_{j \in I_d} \frac{\xi_{ok}^{dj,m}}{ E_{ok}^{dj}},
\ee
which have to be less than 1 for safety.
Similarly, we define critical relative distances as: 
\be
R_{rel}^{dd} = \displaystyle \max_{j\neq l \in I_d} \frac{R_{d}^{d,m}}{R_{dl}^{dj}}, \quad R_{rel}^{ad} = \displaystyle \max_{j \in I_d} \frac{R_{d}^{a,m}}{R_{dj}^{a}},
\ee
which have to be less than 1 for no inter-agent collisions. As observed in Fig.~\ref{fig:crit_rel_dist_and_speeds}, all the critical relative distances are less than 1 and hence there are no collisions.
\begin{figure}[h]
	\centering
	\includegraphics[width=1\linewidth]{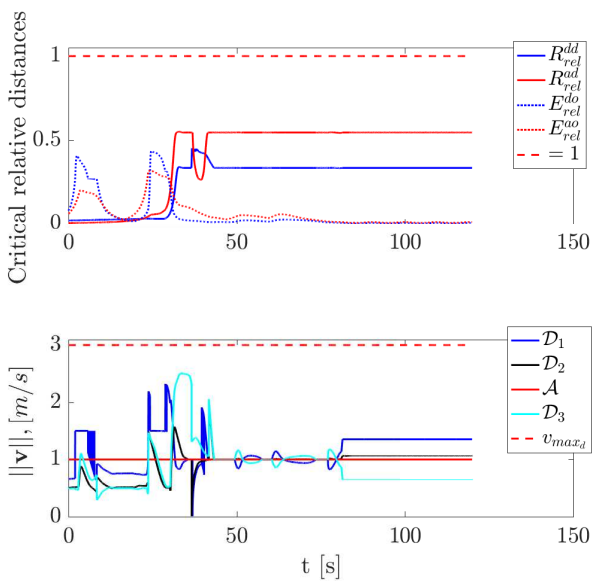}
	\caption{Critical relative distances and control speeds}
	\label{fig:crit_rel_dist_and_speeds}
\end{figure}
Figure \ref{fig:crit_rel_dist_and_speeds} also shows that the defenders' speeds are bounded and less than their maximum speeds.
%
\vspace{-2mm}
\section{Conclusions and Future Work} \label{sec:conclusions}

In this paper, we proposed novel vector-field based guidance laws for defenders herding an attacker to a safe area. Finite-time controllers are designed for the defenders to converge to a formation around the attacker, while safety and convergence are formally proved. Simulation results demonstrate efficacy of the herding strategy. As future work, we plan to investigate the case of multiple attackers under double integrator dynamics.






\section*{ACKNOWLEDGMENT}
We acknowledge the reviewers' valuable comments.
\vspace{-1mm}
\bibliographystyle{IEEEtran}
\bibliography{ACC_2019_Refs}
\end{document}